\documentclass[aps, superscriptaddress]{revtex4-1}

\pdfoutput=1
\usepackage{graphicx}

\usepackage{amsmath}
\usepackage{amssymb}
\usepackage{bm}

\newcommand{\pa}{\partial}

\newcommand{\mean}[1]{\langle{#1}\rangle}

\newcommand{\abs}[1]{{|#1|}}

\usepackage{amsthm}
\newtheorem{Lemma}{Lemma}%[section]
\newtheorem{Definition}{Definition}%[section]
\begin{document}

\title{Asymptotic phase and amplitude for classical and semiclassical stochastic oscillators via Koopman operator theory }

\author{Yuzuru Kato}
\email{Corresponding author: kato.y.bg@m.titech.ac.jp}
\affiliation{Department of Systems and Control Engineering,
	Tokyo Institute of Technology, Tokyo 152-8552, Japan}

\author{Jinjie Zhu}
\affiliation{Department of Systems and Control Engineering,
	Tokyo Institute of Technology, Tokyo 152-8552, Japan}
\affiliation{School of Mechanical Engineering, Nanjing University
	of Science and Technology, Nanjing 210094, China}

\author{Wataru Kurebayashi}
\affiliation{Institute for Promotion of Higher Education, Hirosaki University, Aomori 036-8560, Japan}

\author{Hiroya Nakao}
\affiliation{Department of Systems and Control Engineering,
	Tokyo Institute of Technology, Tokyo 152-8552, Japan}
\date{\today}

\begin{abstract}
	The asymptotic phase is a fundamental quantity for the analysis of deterministic limit-cycle oscillators,
	and generalized definitions of the asymptotic phase for stochastic oscillators have also been proposed.
	In this article, we show that the asymptotic phase and also amplitude can be defined for classical and semiclassical stochastic oscillators in a natural and unified manner by using the eigenfunctions of the Koopman operator of the system.
	We show that the proposed definition gives appropriate values of the phase and amplitude for strongly stochastic limit-cycle oscillators, excitable systems undergoing noise-induced oscillations, and also for quantum limit-cycle oscillators in the semiclassical regime.
\end{abstract}


\maketitle
	
\section{Introduction}

Spontaneous rhythmic oscillations and synchronization are widely observed in various fields of science and technology~\cite{winfree2001geometry, kuramoto1984chemical, pikovsky2001synchronization, nakao2016phase, 
	ermentrout2010mathematical, strogatz1994nonlinear}.
Regular rhythmic oscillations are generally modeled by using nonlinear dynamical systems possessing stable limit-cycle attractors. %MDPI: Is the italics necessary? Please check all -- they are essentially important keywords and we would like to keep them in Italics.

The notion of \textit{asymptotic phase}~\cite{winfree2001geometry, kuramoto1984chemical, pikovsky2001synchronization, nakao2016phase, ermentrout2010mathematical}, which increases with a constant frequency in the basin of the limit-cycle attractor, is a fundamental quantity that provides a basis for \textit{phase reduction}~\cite{winfree2001geometry, kuramoto1984chemical, hoppensteadt1997weakly,	pikovsky2001synchronization, nakao2016phase, ermentrout2010mathematical, strogatz1994nonlinear}, a~standard dimensionality-reduction method for analyzing synchronization of limit-cycle oscillators under the effect of weak perturbation or~coupling.

Recently, the~asymptotic phase and {\it isochrons} (level sets of the asymptotic phase), classical notions in the theory of nonlinear oscillations since Winfree~\cite{winfree1967biological} and Guckenheimer~\cite{guckenheimer1975isochrons}, have been studied from a viewpoint of the Koopman operator theory by Mauroy, Mezi{\'c}, and~Moehlis~\cite{mauroy2013isostables}, and~their relationship with the Koopman eigenfunction associated with the fundamental frequency of the oscillator has been clarified~\cite{mauroy2013isostables, mauroy2020koopman, shirasaka2017phase, kuramoto2019concept, shirasaka2020phase}.
Moreover, they have shown that the (asymptotic) amplitude and {\it isostables}, which characterize deviation of the system state from the limit cycle and extend the Floquet coordinates~\cite{hale1969ordinary2,revzen2012finding,kuramoto2019concept}
to the nonlinear regime, can be introduced naturally in terms of the Koopman eigenfunctions associated with the Floquet exponents with non-zero real parts~\cite{mauroy2013isostables, mauroy2020koopman, shirasaka2017phase, kuramoto2019concept, shirasaka2020phase, kvalheim2021existence}.
By using the asymptotic phase and amplitude functions, we can obtain 
a reduced description of limit-cycle oscillators, which is useful for the analysis and control of synchronization dynamics of limit-cycle oscillators~\cite{mauroy2018global, wilson2016isostable,  monga2019phase, monga2019optimal, zlotnik2013optimal, kato2021optimization, takata2021fast}. 
The theory can also be generalized to delay-differential systems~\cite{kotani2020nonlinear} and spatially extended systems~\cite{nakao2020phase}. 

How to generalize the definition of the conventional asymptotic phase, which was essentially deterministic~\cite{winfree1967biological,guckenheimer1975isochrons}, to~stochastic systems has been an intriguing problem~\cite{teramae2009stochastic, goldobin2010dynamics, nakao2010effective, bonnin2017phase, bonnin2017amplitude, aminzare2019phase, kato2019semiclassical, schwabedal2013phase, thomas2014asymptotic, cao2020partial}.
When the stochasticity is sufficiently weak, the~phase and also amplitude can be defined by using the drift term of the stochastic differential equation (SDE) describing the deterministic vector field of the oscillator.
This approach can also be employed for quantum nonlinear oscillators in the semiclassical regime described by a quantum Fokker--Planck equation (FPE)~\cite{kato2019semiclassical,kato2020semiclassical}.
However, this definition is no longer applicable to strongly stochastic oscillatory systems for which the deterministic vector field does not serve as a clear reference due to the strong effect of~noise.

To cope with this problem, Schwabedal and Pikovsky~\cite{schwabedal2013phase} introduced a definition of the phase in terms of the mean first return time, and~Thomas and Lindner~\cite{thomas2014asymptotic} proposed a definition of the asymptotic phase in terms of the slowest decaying eigenfunction of the backward Fokker--Planck (Kolmogorov) operator describing the mean first passage time, both of which yield phase values that increase with a constant frequency on average for stochastic oscillations in a similar way to the ordinary asymptotic phase for deterministic oscillators.
Recently, we pointed out that the definition of the stochastic asymptotic phase by Thomas and Lindner~\cite{thomas2014asymptotic} can be seen as a natural extension of the deterministic definition from the viewpoint of the Koopman operator theory; namely, it is given by the argument of the Koopman eigenfunction associated with the fundamental frequency~\cite{kato2020quantum} (see also Reference~\cite{engel2021random}) and~extended this idea to the definition of the asymptotic phase for quantum oscillatory~systems.

In this article, based on the Koopman operator theory for stochastic systems, we propose a definition of the asymptotic phase and amplitude  
for strongly stochastic oscillators. 
They are introduced in terms of the eigenfunctions of the Koopman operator associated with the complex eigenvalues with the largest non-zero real part and with the largest non-zero real eigenvalue, respectively, which gives a natural extension of the definition in the deterministic case.
The validity of the proposed definition is illustrated for stochastic limit-cycle oscillations and noise-induced oscillations of excitable systems using noisy Stuart--Landau~\cite{kuramoto1984chemical, nakao2016phase} and FitzHugh--Nagumo~\cite{fitzhugh1961impulses, nagumo1962active} models as examples. 
Moreover, we apply the proposed definition of the stochastic phase and amplitude to a quantum limit-cycle oscillator in the semiclassical regime and show that they also yield appropriate~results.

%%%%%%%%%%%%%%%%%%%%%%%%%%%%%%%%%%%%%%%%%%%%%%%%%%%%%%%%%%%%%%%%%%%%%%%%

\section{Phase and Amplitude for Deterministic Limit-Cycle~Oscillators}
\label{sec:deterministic}
\unskip

\subsection{Classical Definition of the Asymptotic Phase and~Amplitude}

In this section, we review the definition of the asymptotic phase and amplitude for deterministic limit-cycle oscillators and discuss their relationship with the Koopman eigenfunctions~\cite{mauroy2013isostables, mauroy2020koopman,wilson2016isostable, shirasaka2017phase, kuramoto2019concept, nakao2020phase, shirasaka2020phase}.
We consider a deterministic dynamical system 
\begin{align}
	\dot {\bm X}(t) = \bm{A}(\bm{X}(t)),
	\label{system0}
\end{align}
where $\bm{X}(t) \in \mathbb R^{N}$ is the system state at time $t$, ${\bm A} : {\mathbb R}^{N} \to {\mathbb R}^{N}$ is a sufficiently smooth vector field representing the system dynamics, and~the dot $(\dot{})$ represents the time derivative.
We assume that the system has an exponentially stable limit-cycle solution ${\bm{X}}_{0}(t)$ with a natural period $T$ and frequency $\omega = 2\pi / T$, satisfying ${\bm X}_0(t+T) = {\bm X}_0(t)$. 
We denote this limit cycle as $\chi $ and its basin of attraction as $B_\chi \subseteq  {\mathbb R}^{N}$.
Instead of the time $t$, we can parameterize a point on the limit cycle $\chi$ using a phase $\phi \in [0, 2\pi)$ as ${\bm \chi}(\phi) = {\bm X}_0(\omega t)$ ($0 \leq t < T$), where the phase value $\phi = \omega t$ increases linearly with time $t$ from $0$ to $2\pi$ ($2\pi$ is identified with $0$), and the origin of the phase $\phi = 0$ is assigned to the state ${\bm X}_0(0)$ without loss of~generality.

The linear stability of $\chi$ is characterized by the Floquet exponents $\lambda_j \in {\mathbb C}$\linebreak ($j=0, 1, ..., N-1$)~\cite{hale1969ordinary2,guckenheimer1982nonlinear},
which we sort in decreasing order of their real parts, i.e.,~$\mbox{Re} (\lambda_0) \geq \mbox{Re}({\lambda_1}) \geq \mbox{Re}({\lambda_2}) \geq \cdots \geq \mbox{Re}({\lambda_{N-1}})$.
Here, the~exponent $\lambda_0$ is zero and associated with the phase direction tangent to $\chi$, and~the other exponents $\lambda_1, \dots , \lambda_{N-1}$ possess negative real parts because $\chi$ is exponentially stable and is associated with the amplitude directions deviating from $\chi$.
We further assume that $\lambda_1$ is real and $\lambda_1 \gg \mbox{Re}(\lambda_2)$, namely, the~relaxation of the slowest decaying mode is non-oscillatory and much slower than the other faster decaying modes.
Such a situation often occurs in realistic models of limit-cycle oscillators. We can then focus only on the slowest decaying mode and introduce a single real amplitude associated with~it.

The asymptotic phase function $\Phi_0 : B_\chi \to [0, 2\pi)$ and amplitude function $R_0 : B_\chi \to {\mathbb R}$ of the limit cycle $\chi$ are defined in the basin $B_{\chi}$ of $\chi$ such that
\begin{align}
	{\bm A}({\bm{X}})  \cdot  \nabla \Phi_0({\bm{X}}) &= \omega, \cr
	{\bm A}({\bm{X}})  \cdot  \nabla R_0({\bm{X}}) &=  \lambda_1 R_0({\bm X}),
	\label{defphaseamp}
\end{align}
are satisfied for all ${\bm X} \in B_\chi$~\cite{kuramoto2019concept}.
Here, the~inner product is defined as $ \bm{a} \cdot \bm{b}  = \sum_{j = 1}^{N} \overline{ a_{j} } b_{j}$ (the overline denotes complex conjugate) and $\nabla = \partial / \partial {\bm X}$ represents the gradient with respect to ${\bm X}$.
As stated above, we focus only on the phase associated with $\lambda_0=0$ and the slowest decaying amplitude associated with real $\lambda_1$.
In general, we can introduce $N-1$ amplitude variables associated with $N-1$ exponents $\lambda_1, \dots , \lambda_{N-1}$, which are in general complex, and~obtain a closed set of %phase-amplitude 
equations for the phase and amplitudes~\cite{kuramoto2019concept}.
The level sets of the phase function are called {\it isochrons}~\cite{winfree1967biological,guckenheimer1975isochrons} and those of the amplitude function are called {\it isostables}~\cite{mauroy2013isostables}.

By using the above definition, we can introduce the phase and amplitude variables for the oscillator state ${\bm X}(t) \in B_\chi$ at time $t$ as $\phi(t) = \Phi_0({\bm X}(t))$ and $r(t) = R_0({\bm X}(t))$, which~obey
\begin{align}
	\dot{\phi}(t) = \dot{\Phi}_0({\bm X}(t)) &= {\bm A}({\bm X}(t)) \cdot \nabla \Phi_0({\bm X}(t)) = \omega, \cr
	\dot{r}(t) = \dot{R}_0({\bm X}(t)) &= {\bm A}({\bm X}(t)) \cdot \nabla R_0({\bm X}(t)) =  \lambda_1 R_0({\bm X}) = \lambda_1 r(t),
	\
	\label{phase_amp}
\end{align}
that is, the~phase $\phi$ always increases with a constant frequency $\omega$ and the amplitude $r$ decays exponentially with the rate $\lambda_1$ as ${\bm X}$ evolves in $B_\chi$ toward $\chi$.

Note that the phase function is determined only up to an arbitrary constant and the scale of the amplitude function $R_0({\bm X})$ is also arbitrary, because~$\Phi_0({\bm{X}}) + c_1$ and $c_2 R_0({\bm X})$ with arbitrary constants $c_1, c_2 \in {\mathbb R}$ also satisfy Equation~(\ref{defphaseamp}).
Suppose that the initial state is ${\bm X}_0 \in B_\chi$ at $t=0$. If~we assign the phase $\phi(0;{\bm X}_0)$ and amplitude $r(0;{\bm X}_0)$ to the initial state ${\bm X}_0$, we obtain $\phi(t;{\bm X}_0) = \omega t + \phi(0;{\bm X}_0)$ and $r(t;{\bm X}_0) = r(0;{\bm X}_0) \exp(\lambda_1 t)$, whose
dependence on ${\bm X}_0$ is explicitly~shown.

By focusing only on the asymptotic phase and amplitude, we can perform {\it phase-amplitude reduction} (or isochron-isostable reduction) of a limit-cycle oscillator~\cite{wilson2016isostable, shirasaka2017phase, mauroy2018global, shirasaka2020phase}, in~which we reduce the dimensionality of the system dynamics from $N$ to $2$ and approximately 
describe it by a simple set of two-dimensional phase and amplitude equations.
The phase equation has been extensively used for the analysis of weakly coupled limit-cycle oscillators~\cite{winfree2001geometry, kuramoto1984chemical, pikovsky2001synchronization, nakao2016phase, ermentrout2010mathematical, strogatz1994nonlinear}, and~the amplitude equation has also been used recently for the analysis and control of limit-cycle oscillators~\cite{mauroy2018global, wilson2016isostable,  monga2019phase, monga2019optimal, takata2021fast}. 

\subsection{Koopman Operator~Viewpoint}

The asymptotic phase and amplitude introduced in the previous subsection are closely related to the Koopman operator of the system~\cite{mauroy2018global, wilson2016isostable,  monga2019phase, monga2019optimal}.
The Koopman operator $U^{\tau}$, which describes the evolution of a general {\it observable} $g$ of the system state  ${\bm X} \in {\mathbb R}^N$, is defined as
\begin{align}
	(U^{\tau} g)({\bm X}) = g(S^{\tau} {\bm X}),
\end{align}
where $g : {\mathbb R}^N \to {\mathbb C}$ is the observable, i.e.,~an observation function that maps a system state to an observed value, and~$S^{\tau} : {\mathbb R}^N \to {\mathbb R}^N$ is a flow of the system satisfying ${\bm X}(t+\tau) = S^{\tau} {\bm X}(t)$ for $\tau \geq 0$.
When the flow $S^{\tau}$ is analytic, it can be expanded as
\begin{align}
	S^{\tau} {\bm X} = {\bm X} + \tau {\bm A}({\bm X}) + O(\tau^2)
\end{align}
for $|\tau| \ll 1$.
Considering an analytic observable $g$, we can expand it as
\begin{align}
	g(S^{\tau} {\bm X}) = g({\bm X}) + \tau {\bm A}({\bm X}) \cdot \nabla g({\bm X}) + O(\tau^2).
\end{align}
Therefore, the~infinitesimal time evolution of $g$ can be expressed as
\begin{align}
	\frac{d}{dt} g({\bm X}) = \lim_{\tau \to 0} 
	\frac{ U^{\tau} g({\bm X}) - g({\bm X}) }{\tau} = \lim_{\tau \to 0} \frac{ g(S^{\tau} {\bm X}) - g({\bm X}) }{\tau} = {\bm A}({\bm X}) \cdot \nabla g({\bm X}).
\end{align}
The operator
\begin{align}
	{A} = {\bm A}({\bm X}) \cdot \nabla,
	\label{koopmaninfini}
\end{align}
which appeared in Equation~(\ref{phase_amp}), can thus be interpreted as an infinitesimal generator of the Koopman operator $U^\tau$.

For the limit-cycle oscillator described by Equation~(\ref{system0}), we can easily confirm that the complex exponential of the phase function $\Phi_0({\bm X})$,
\begin{align}
	\Psi_0({\bm X}) = e^{i \Phi_0({\bm X})},
\end{align}
is an eigenfunction of the operator ${A}$ with an eigenvalue $i \omega_1$ where $i = \sqrt{-1}$, because~%
\begin{align}
	{A} \Psi_0({\bm X}) = i \omega \Psi_0({\bm X})
	\label{eq:detiw}
\end{align}
is satisfied for ${\bm X} \in B_{\chi}$.
Therefore, from~the viewpoint of the Koopman operator theory, the~asymptotic phase can be introduced as the argument (polar angle) of the Koopman eigenfunction $\Psi_0({\bm X})$ associated with the eigenvalue $i \omega$, which is determined by the natural frequency $\omega$ of the oscillator~\cite{mauroy2018global, wilson2016isostable,  monga2019phase, monga2019optimal},
as
\begin{align}
	\label{eq:koopmaniw}
	\Phi_0({\bm X}) = \mbox{\rm Arg}\ \Psi_0({\bm X}),
\end{align}
where $\mbox{\rm Arg}$ represents the principal argument of a complex number in the range $[0, 2\pi)$.
Moreover, the~asymptotic amplitude function $R_0({\bm X})$ is nothing but the eigenfunction of the linear operator ${A}$ associated with the eigenvalue $ \lambda_1 $ for ${\bm X} \in B_\chi$, i.e.,
\begin{align}
	{A} R_0({\bm X}) = \lambda_1 R_0({\bm X}).
	\label{eq:koopmanamp}
\end{align}

Thus, the~Koopman operator theory provides a natural and unified definition of the asymptotic phase and amplitude, and~the simplified Equation~(\ref{phase_amp}) in the phase-amplitude coordinates can be interpreted as a global linearization of the nonlinear dynamics of the limit-cycle oscillator by using the Koopman eigenfunctions.
In References~\cite{mauroy2012use, mauroy2018global}, Mauroy and Mezi\'{c} pointed out these facts and explicitly calculated the phase and amplitude functions for several models of limit-cycle~oscillators.

%%%%%%%%%%%%%%%%%%%%%%%%%%%%%%%%%%%%%%%%%

\section{Fokker--Planck Equation and Stochastic Koopman~Operator}
\unskip

\subsection{Forward and Backward Fokker--Planck~Equations}

In the previous section, we considered deterministic limit-cycle oscillators and introduced the asymptotic phase and amplitude functions from the Koopman-operator viewpoint. Our aim in this study is to generalize the idea to stochastic oscillatory systems. In~this section, we review some basic facts on the Fokker--Planck equations and stochastic Koopman operator for stochastic dynamical~systems.

We consider a stochastic dynamical system described by a time-homogeneous SDE of Ito type~\cite{arnold1974stochastic,gardiner2009stochastic,pavliotis2014stochastic},
\begin{align}
	d{\bm X}(t) = {\bm A}({\bm X}(t)) dt + {\bm B}({\bm X}(t)) d{\bm W}(t),
	\label{sde}
\end{align}
where $\bm{X}(t) \in \mathbb R^{N}$ is the system state at time $t$, ${\bm A} : {\mathbb R}^N \to {\mathbb R}^N$ is a drift term representing the deterministic vector field of the oscillator, ${\bm B} : {\mathbb R}^N \to {\mathbb R}^{N \times N}$ is a matrix characterizing the intensity of the noise, and~${\bm W}(t)$ is a Wiener process in ${\mathbb R}^N$ representing the $N$-dimensional independent Gaussian-white noise.
We assume that ${\bm A}$ and ${\bm B}$ satisfy the Lipschitz condition $| {\bm A}({\bm X}) - {\bm A}({\bm Y}) | + | {\bm B}({\bm X}) - {\bm B}({\bm Y}) | \leq K | {\bm X} - {\bm Y} |$ and the growth condition $|{\bm A}({\bm X})|^2 + |{\bm B}({\bm X})|^2 \leq K^2 | ( 1 + |{\bm X}|^2 )$ with some constant $K$ for Equation~(\ref{sde}) to possess a unique strong solution ${\bm X}(t)$~\cite{arnold1974stochastic,pavliotis2014stochastic}.

The FPE equivalent to the above SDE, describing the time evolution of the probability density function (PDF) $p({\bm X}, t) : {\mathbb R}^N \times {\mathbb R} \to {\mathbb R}$ of ${\bm X}$ at time $t$ is given by 
\begin{align}
	\label{eq:fpe}
	\frac{\pa}{\pa t} p({\bm X}, t) = {L}_{\bm X} p({\bm X}, t)
	= \left[ - \frac{\pa}{\pa {\bm X}} {\bm A}({\bm X}) + \frac{1}{2} \frac{\pa^2}{\pa {\bm X}^2} {\bm D}({\bm X}) \right] p({\bm X}, t),
\end{align}
where
\begin{align}
	\label{eq:fpe2}
	{L}_{\bm X} = - \frac{\pa}{\pa {\bm X}} {\bm A}({\bm X}) + \frac{1}{2} \frac{\pa^2}{\pa {\bm X}^2} {\bm D}({\bm X})
\end{align}
is a (forward) Fokker--Planck operator.
Here, the~drift vector ${\bm A} : {\mathbb R}^N \to {\mathbb R}^N$ is the same as in Equation~(\ref{sde})
and ${\bm D} = \bm{B} \bm{B}^{\sf T} : {\mathbb R}^N \to {\mathbb R}^{N \times N}$
is a symmetric diffusion matrix, where ${\sf T}$ indicates matrix transposition.
We also assume that the functions ${\bm A}({\bm X})$ and ${\bm D}({\bm X})$ are smooth, satisfy the growth conditions $| {\bm D}({\bm X}) | \leq M$, $|- {\bm A}({\bm X}) + \nabla \cdot {\bm D}({\bm X}) | \leq M (1 + | {\bm X} | )$, and~$|- \nabla \cdot {\bm A}({\bm X}) + (1/2) \nabla \cdot (\nabla \cdot {\bm D}({\bm X}) ) | \leq M (1 + | {\bm X} |^2 )$ with some constant $M$, and~the uniform parabolicity $({\bm \lambda} \cdot {\bm D}( {\bm X}) {\bm \lambda}) \geq \alpha | {\bm \lambda} |^2$ for all ${\bm \lambda} \in {\mathbb R}^N$ with a constant $\alpha > 0$
in order that Equation~(\ref{eq:fpe}) possesses
a classical solution for $t>0$~\cite{arnold1974stochastic,lasota2008probabilistic,lasota2013chaos,pavliotis2014stochastic}.

The transition probability density $p(\bm {X}, t| \bm {Y}, s)$, satisfying $p({\bm X}, t) =$ 
	$\int p(\bm {X}, t| \bm {Y}, s)$ $p({\bm Y}, s) d{\bm Y}$ for $t > s$ and $\lim_{t \to s+0} p({\bm Y}, t | {\bm X}, s) = \delta({\bm X} - {\bm Y})$ where $\delta({\bm X}-{\bm Y})$ is Dirac's delta measure~\cite{arnold1974stochastic,gardiner2009stochastic,pavliotis2014stochastic}, obeys the forward FPE
\begin{align}
	\frac{\pa}{\pa t} p(\bm {X}, t| \bm {Y}, s) = {L}_{\bm X} p(\bm {X}, t| \bm {Y}, s)
	\label{forward}
\end{align}
and also the corresponding backward FPE
\begin{align}
	\frac{\pa}{\pa s} p(\bm {X}, t| \bm {Y}, s) &= -{L}^{+}_{\bm Y} p(\bm {X}, t| \bm {Y}, s) 
	=- \left[ {\bm A}({\bm Y}) \frac{\pa}{\pa {\bm Y}} + \frac{1}{2} {\bm D}({\bm Y}) \frac{\pa^2}{\pa {\bm Y}^2}  \right] p(\bm {X}, t| \bm {Y}, s).
	\label{bkfpe}
\end{align}

Here, the~backward Fokker--Planck operator
\begin{align}
	{L}_{\bm X}^{+} = {\bm A}({\bm X}) \frac{\pa}{\pa {\bm X}} + \frac{1}{2} {\bm D}({\bm X}) \frac{\pa^2}{\pa {\bm X}^2} 
	\label{eq:Lxadj}
\end{align}
is the adjoint linear operator of ${L}_{\bm X}$
with respect to the $L^2$ inner product
\begin{align}
	\mean{G(\bm{X}), H(\bm{X})} = \int \overline{ G(\bm{X}) } H(\bm{X}) d \bm{X}
\end{align}
of two functions $G, H : {\mathbb R}^N \to {\mathbb C}$, i.e.,~
\begin{align}
	\mean{ {L}^{+}_{\bm{X}} G(\bm{X}), H(\bm{X})}
	=\mean {G(\bm{X}), {L}_{\bm X} H(\bm{X})},
\end{align}
where the overline indicates complex conjugate and the integration is taken over the whole range of ${\bm X}$ here and~hereafter. 

\subsection{Eigensystem of the Fokker--Planck~Operators}

The linear differential operators ${L}_{\bm X}$ and ${L}_{\bm X}^{+}$ have the eigensystem $\{\Lambda_{k}, P_{k}, \overline{Q_{k}}\}_{k\geq 0}$ consisting of the eigenvalue $\Lambda_{k}$ and eigenfunctions $P_k({\bm X})$, $\overline{Q_k}({\bm X})$ satisfying
\begin{align}
	&{L}_{\bm X} P_{k}(\bm{X}) = \Lambda_k P_{k}(\bm{X}),
	\cr
	&{L}_{\bm X}^{+} \overline{ Q^{}_{k} }(\bm{X}) =  \Lambda_k \overline{ Q_{k}} (\bm{X}),
\end{align}
and the biorthogonality conditions
\begin{align}
	&\mean{ Q_{k}(\bm{X}), P_{l}(\bm{X})} = \delta_{kl},
\end{align}
where $k, l = 0, 1, 2, \ldots$ and $\delta_{kl}$ represents Kronecker's delta~\cite{gardiner2009stochastic,risken1996fokker,thomas2014asymptotic}.
Here, $Q_k({\bm X})$ is the complex conjugate of $\overline{Q_k}({\bm X})$, which is an eigenfunction of ${L}_{\bm X}^{+}$ 
associated with the eigenvalue $\overline{\Lambda_k}$, i.e.,~${L}_{\bm X}^{+} Q^{}_{k}(\bm{X}) = \overline{ \Lambda_k } Q_{k}(\bm{X})$.
Because $L_{\bm X}$ is a Fokker--Planck operator, the~eigenvalue $\Lambda_0$ is zero and the associated eigenfunction $P_0({\bm X})$ gives the stationary PDF of the FPE, i.e.,~$L_{\bm X} P_0({\bm X}) = 0$, when appropriately normalized. All other eigenvalues have negative real parts and the associated eigenfunctions represent the relaxation eigenmodes of the FPE that eventually decay as $t \to \infty$~\cite{gardiner2009stochastic,risken1996fokker,thomas2014asymptotic}.

\subsection{Stochastic Koopman~Operator}

We here introduce the stochastic Koopman operator following Mezi\'c~\cite{mezic2005spectral} and discuss its relationship with the backward Fokker--Planck~operator.

\begin{Definition} For an observable $g : {\mathbb R}^{N} \to {\mathbb C}$ and $\tau > 0$, the~stochastic Koopman operator $U_{st}^{\tau}$ is defined as
	\begin{align}
		U_{st}^{\tau} g({\bm X}) = {\mathbb E} [ g(S_{st}^{\tau} {\bm X}) ] = \int p({\bm Y}, \tau | {\bm X}, 0) g({\bm Y}) d{\bm Y}
		\label{defstochastickoopman}
	\end{align}
	for ${\bm X} \in {\mathbb R}^N$, where ${\mathbb E} [ \cdot ]$ represents the expectation over realizations of the stochastic flow $S_{st}^{\tau}$ of Equation~(\ref{sde}) and $p({\bm Y}, \tau | {\bm X}, 0)$ is the transition probability density satisfying
	Equation~(\ref{forward}).
\end{Definition}

In the second expression of Equation~(\ref{defstochastickoopman}), the~expectation ${\mathbb E} [ g(S_{st}^{\tau} {\bm X}) ]$ is represented as an average over the transition probability density $p({\bm Y}, \tau | {\bm X}, 0)$. The~initial time can be taken as $0$ without loss of generality because the process is time-homogeneous.
We also introduce the infinitesimal generator of the stochastic Koopman~operator.

\begin{Definition}
	For an observable $g : {\mathbb R}^{N} \to {\mathbb C}$, the~infinitesimal generator $A_{st}$ of the stochastic Koopman operator $U_{st}^{\tau}$ is defined by
	\begin{align}
		A_{st} g({\bm X}) &= \lim_{\tau \to +0} \frac{ U_{st}^{\tau} g({\bm X}) - g({\bm X}) } {\tau}.
	\end{align}
\end{Definition}

From the above definitions, it can be shown that the infinitesimal generator of the stochastic Koopman operator is given by the backward Fokker--Planck~operator.

\begin{Lemma}
	The infinitesimal generator $A_{st}$ of the stochastic Koopman operator $U_{st}^\tau$ 
	%in \mbox{Equation~(\ref{defstochastickoopman})} 
	is given by the backward Fokker--Planck operator $L_{\bm X}^+$ in Equation~(\ref{eq:Lxadj}).
\end{Lemma}

The proof can be found in the textbook by {\O}ksendal~\cite{oksendal2000stochastic} (Section 7.3, The generator of an Ito diffusion, Theorem 7.3.3).\\

Thus, the~infinitesimal generator of the stochastic Koopman operator is given by the backward Fokker--Planck operator, i.e.,~$A_{st} = L_{\bm X}^+$. 
Before proceeding to the definition of the asymptotic phase and amplitude, we show a result on the time evolution  of the average of the eigenfunction $\overline{Q_k}$ $(k=0, 1, 2, \cdots)$ of $A_{st} = L_{\bm X}^+$.

\begin{Lemma}
	Let ${\bm X}(t) = S_{st}^t {\bm X}_0$ be a solution to
	Equation~(\ref{sde}) with
	an initial condition ${\bm X}_0 \in {\mathbb R}^N$, where $S_{st}^t : {\mathbb R}^N \to {\mathbb R}^N$ ($t \geq 0$) is the stochastic flow of Equation~(\ref{sde}). Then, the~average
	\begin{align}
		\mathbb{E}[  \overline{Q_k}(S_{st}^t {\bm X}_0)  ]
		=
		\int \overline{Q_k}({\bm X})  p({\bm X}, t | {\bm X}_0, 0) d{\bm X}
	\end{align}
	of $\overline{Q_k}({\bm X}(t)) =  \overline{Q_k}(S_{st}^t {\bm X}_0)$ obeys
	\begin{align}
		\frac{d}{dt} 
		\mathbb{E}[  \overline{Q_k}(S_{st}^t {\bm X}_0)  ]
		=
		\Lambda_k
		\mathbb{E}[  \overline{Q_k}(S_{st}^t {\bm X}_0)  ]
	\end{align}
	for arbitrary ${\bm X}_0$, where ${\mathbb E} [ \cdot ]$ represents the expectation over realizations of the stochastic flow $S_{st}^{t}$ and $p({\bm X}, t | {\bm X}_0, 0)$ is the transition probability density satisfying
	Equation~(\ref{forward}).
\end{Lemma}

\begin{proof}
	\begin{align}
		&
		\frac{d}{dt} \mathbb{E}[  \overline{Q_k}(S_{st}^t {\bm X}_0)  ]
		= \frac{d}{dt} \int {\overline{Q_k}({\bm X})} p({\bm X}, t| {\bm X}_0, 0) d{\bm X}
		\cr
		&= \int {\overline{Q_k}({\bm X})} \frac{\partial}{\partial t} p({\bm X}, t| {\bm X}_0, 0) d{\bm X}
		= \int {\overline{Q_k}({\bm X})} L_{\bm X} p({\bm X}, t| {\bm X}_0, 0) d{\bm X} 
		\cr 
		&	= \int L_{\bm X}^+ {\overline{Q_k}({\bm X})} p({\bm X}, t| {\bm X}_0, 0) d{\bm X}
		= \int \Lambda_k \overline{Q_k}({\bm X}) p({\bm X}, t| {\bm X}_0, 0) d{\bm X}
		\cr
		&=  \Lambda_k \mathbb{E}[  \overline{Q_k}(S_{st}^t {\bm X}_0)  ].
	\end{align}
\end{proof}

We use the above result for discussing the evolution of the averaged phase and amplitude in the next~section.

\section{Phase and Amplitude for Stochastic Oscillatory~Systems}
\label{sec:stochastic}
\unskip

\subsection{Stochastic Oscillatory~Systems}

The definitions of the phase and amplitude in Section~\ref{sec:deterministic} are based on the deterministic limit-cycle solution.
These definitions are still applicable to noisy limit-cycle oscillators when the noise can be regarded as a weak perturbation~\cite{winfree2001geometry, kuramoto1984chemical, pikovsky2001synchronization, nakao2016phase, ermentrout2010mathematical}.
However, they are no longer valid when the oscillator is subjected to stronger noise because we cannot rely on the deterministic limit-cycle solution in defining the phase and amplitude~functions.

In Reference~\cite{thomas2014asymptotic}, Thomas and Lindner proposed a definition of the asymptotic phase for strongly stochastic oscillators without relying on the limit-cycle solution of the deterministic system, where they used the slowest decaying eigenfunction of the backward Fokker--Planck operator as the phase function based on the consideration of the mean first passage time.
In this section, we show that their definition can be viewed as a natural extension of the deterministic definition in the sense that it is given by the argument of the Koopman eigenfunction associated with the fundamental frequency~\cite{engel2021random,kato2020quantum}.

\subsection{Assumptions on the~Eigenvalues}

Since we consider \textit{oscillatory} stochastic systems, we introduce the following assumptions on the eigenvalues of the Fokker--Planck operator $L_{\bm X}$ in Equation~(\ref{eq:fpe2}).

(i) We assume that the eigenvalues with the largest non-zero real part are given by a complex-conjugate pair, i.e.,~the slowest decaying eigenmode is oscillatory, and~regard this eigenmode as the fundamental oscillation of the system.
These eigenvalues are represented~	as
\begin{align}
	\Lambda_1 = \mu_1 + i \omega_1,
	\quad
	\overline{ \Lambda_1 } = \mu_1 - i \omega_1,
\end{align}
where $\mu_1 < 0$ and $\omega_1 > 0$ characterize the decay rate and {\it fundamental frequency} of the %fundamental 
oscillation, respectively.

(ii) We assume that the largest non-zero eigenvalue on the real axis, denoted as $\Lambda_2$, is smaller than $\mu_1$, i.e.,~$\Lambda_2 < \mbox{Re}~\Lambda_1$, and~consider that this eigenvalue characterizes the {\it decay rate} of the amplitude of the system, i.e.,~the deviation of the system state from the averaged oscillatory~state.

When the system has a stable limit cycle as discussed in the previous section in the limit of vanishing noise intensity, these eigenvalues are expected to converge to $i \omega$ and 
$\lambda_1$ of the deterministic limit cycle in the limit of vanishing noise, i.e.,~$\omega_1 \to \omega$ and $\Lambda_2 \to \lambda_1$.

\subsection{Definition of the Asymptotic Phase~Function}

Thomas and Lindner~\cite{thomas2014asymptotic} defined an asymptotic phase 
for the stochastic oscillatory system, Equation~(\ref{sde}),
by using the argument of the (complex conjugate of the) eigenfunction 
$\overline{Q_1}({\bm X})$ of the backward Fokker--Planck operator ${L}_{\bm X}^{+}$ associated with the eigenvalue $\Lambda_1$ characterized by the fundamental frequency $\omega_1$ (in the notation of the present study), satisfying $L_{\bm X}^+ \overline{Q_1}({\bm X}) = \Lambda_1 \overline{Q_1}({\bm X})$, as~
\begin{align}
	\Phi(\bm{X}) = \mbox{\rm Arg}\  \overline{Q_1}({\bm X}),
	\label{eq:TLphase}
\end{align}
and showed that this $\Phi({\bm X})$ gives an appropriate phase value that increases with a constant frequency $\omega_1$ with the evolution of ${\bm X}$ on average. 
They showed that, in~the limit of vanishing noise where the system is described by the vector field ${\bm A}({\bm X})$ possessing a stable limit-cycle solution, this definition of the asymptotic phase coincides with the deterministic definition in Section~\ref{sec:deterministic}~\cite{thomas2014asymptotic}.
In Reference~\cite{kato2020quantum}, we pointed out that the above definition of the phase function by the backward Fokker--Planck operator can also be understood from the viewpoint of the Koopman operator theory.
In what follows, we introduce the asymptotic phase and also the amplitude from the Koopman-operator~viewpoint.\\

Let us rephrase the above definition of the asymptotic phase ${\Phi}({\bm X})$ for stochastic oscillators 
from the Koopman-operator~viewpoint.

\begin{Definition}
	We define the asymptotic phase ${\Phi}({\bm X})$ of the oscillator state ${\bm X} \in {\mathbb R}^N$ described by Equation~(\ref{sde}) by using the eigenfunction $\overline{ Q_1 }(\bm{X})$ of the infinitesimal generator of the Koopman operator $A_{st} = L_{\bm X}^+$ in Equation~(\ref{eq:Lxadj}) associated with the eigenvalue $\Lambda_1$ as 
	\begin{align}
		\Phi(\bm{X}) = \mbox{\rm Arg}\ \overline{ Q_1  }(\bm{X}).
		\label{eq:TLphase2}
	\end{align}
\end{Definition}

Note that the above phase $\Phi({\bm X})$ has a discontinuity at $2\pi$, which causes difficulty in taking ensemble averages of $\Phi({\bm X})$ over realizations of ${\bm X}$.
Rather, as~in the standard convention in directional statistics~\cite{fisher1995statistical}, we consider a `wrapped' distribution of the phase values and use the circular mean to calculate the average phase.
This is accomplished by taking the ensemble average of $\overline{Q_1}({\bm X})$ over many realizations and then calculate its argument, rather than calculating the ensemble average of $ \mbox{\rm Arg}\ \overline{ Q_1  }(\bm{X})$.

\begin{Definition}
	We define the averaged asymptotic phase of the stochastic oscillator described by Equation~(\ref{sde}) at time $t$, started from an initial condition ${\bm X}_0$ at time $0$, as~%
	\begin{align}
		\phi(t ; {\bm X}_0)
		=
		\mbox{\rm Arg} \ \mathbb{E}[  \overline{Q_1}(S_{st}^t {\bm X}_0)  ]
		=
		\mbox{\rm Arg} \int \overline{Q_1}({\bm X})  p({\bm X}, t | {\bm X}_0, 0) d{\bm X},
		\label{avgasymphase}
	\end{align}
	where ${\mathbb E} [ \cdot ]$ represents the expectation over realizations of the stochastic flow $S_{st}^{t}$ of Equation~(\ref{sde}) and $p({\bm X}, t | {\bm X}_0, 0)$ is the transition probability density satisfying
	Equation~(\ref{forward}).
\end{Definition}

Let us confirm that the above definition of the phase function 
%using the Koopman eigenfunction $ \overline{Q_1}({\bm X})$ 
yields appropriate phase values on average. 

\begin{Lemma}
	The average asymptotic phase in Equation~(\ref{avgasymphase}) increases with a constant frequency $\omega_1$, i.e.,
	\begin{align}
		\frac{d}{dt} \phi(t ; {\bm X}_0) 
		= \omega_1,
		\label{lemma2}
	\end{align}
	for arbitrary ${\bm X}_0 \in {\mathbb R}^N$.
\end{Lemma}

\begin{proof}
	From Lemma 2, the~average $\mathbb{E}[  \overline{Q_1}(S_{st}^t {\bm X}_0)  ]$ of
	$\overline{Q_1}({\bm X}(t)) = \overline{Q_1}(S_{st}^t {\bm X}_0)$ obeys
	\begin{align}
		\frac{d}{dt}\mathbb{E}[  \overline{Q_1}(S_{st}^t {\bm X}_0)  ]
		= \Lambda_1 \mathbb{E}[  \overline{Q_1}(S_{st}^t {\bm X}_0)  ],
	\end{align}
	where $\Lambda_1 = \mu_1 + i \omega_1$.
	By integration, we obtain
	\begin{align}
		\mathbb{E}[  \overline{Q_1}(S_{st}^t {\bm X}_0)  ] = e^{(\mu_1 + i \omega_1)t} \mathbb{E}[  \overline{Q_1}( {\bm X}_0)  ],
	\end{align}
	where we used that $S^0_{st} {\bm X}_0 = {\bm X}_0$. The~averaged asymptotic phase is thus given by
	\begin{align}
		\phi(t ; {\bm X}_0) = \mbox{\rm Arg}\ \mathbb{E}[  \overline{Q_1}(S_{st}^t {\bm X}_0)  ] 
		=
		\omega_1 t + \mbox{\rm Arg}\ \mathbb{E}[  \overline{Q_1}({\bm X}_0) ]
		=
		\omega_1 t + \phi(0 ; {\bm X}_0),
	\end{align}
	which yields Equation~(\ref{lemma2}) by differentiation
	by $t$.
\end{proof}

Thus, the~averaged asymptotic phase $\phi(t; {\bm X}_0)$ of the oscillator satisfies
\begin{align}
	\dot\phi(t; {\bm X}_0) = \omega_1,
\end{align}
namely, $\phi(t; {\bm X}_0)$ increases with a constant frequency $\omega_1$ on average for any ${\bm X}_0$.
This result indicates that the definition of the asymptotic phase in Equation~(\ref{eq:TLphase}) for the stochastic oscillators by Thomas and Lindner~\cite{thomas2014asymptotic} is a natural extension of the definition in \mbox{Equation~(\ref{eq:koopmaniw})} for the deterministic oscillators from the Koopman-operator~viewpoint.

\subsection{Definition of the Amplitude~Function}

We have seen that the definition of the stochastic asymptotic phase by using the backward Fokker--Planck operator can be naturally interpreted as a generalization of the deterministic definition from the viewpoint of the Koopman operator theory.
Furthermore, as~explained in Section~\ref{sec:deterministic}, the~asymptotic amplitude can be naturally defined by using the Koopman eigenfunction associated with the largest non-zero real eigenvalue in deterministic systems.
Therefore, to~generalize the definition of the amplitude to stochastic oscillators, it appears natural to use the eigenfunction of the stochastic Koopman~operator.

\begin{Definition}
	We define the amplitude of the oscillator state ${\bm X} \in {\mathbb R}^N$ described by Equation~(\ref{sde}) by using the Koopman eigencfunction $\overline{ Q_2  }(\bm{X})$ of $L_{\bm X}^+$ in Equation~(\ref{eq:Lxadj}) associated with the largest non-zero real eigenvalue $\Lambda_2$ as 
	\begin{align}
		R(\bm{X}) =  \overline{ Q_2  }(\bm{X}).
		\label{eq:TLamp}
	\end{align}
\end{Definition}

Let us confirm that the above definition of the amplitude function %using the Koopman eigenfunction $R({\bm X}) = \overline{Q_2}({\bm X})$ 
yields appropriate amplitude values on~average. 

\begin{Definition}
	We define the averaged amplitude of the stochastic oscillator described by \mbox{Equation~(\ref{sde})} at time $t$, started from the initial condition ${\bm X}(0) = {\bm X}_0$ at time $0$, as~
	\begin{align}
		r(t ; {\bm X}_0)
		=
		\mathbb{E}[  \overline{Q_2}(S_{st}^t {\bm X}_0)  ]
		=
		\int \overline{Q_2}({\bm X})  p({\bm X}, t | {\bm X}_0, 0) d{\bm X},
		\label{avgasymamp}
	\end{align}
	where ${\mathbb E} [ \cdot ]$ represents the expectation over realizations of the stochastic flow $S_{st}^{t}$ of Equation~(\ref{sde}) and $p({\bm X}, t | {\bm X}_0, 0)$ is the transition probability density satisfying
	Equation~(\ref{forward}).
\end{Definition}

\begin{Lemma}
	The averaged amplitude $r(t ; {\bm X}_0)$ in Equation~(\ref{eq:TLamp}) decays at a constant rate $\Lambda_2$, i.e.,
	\begin{align}
		\frac{d}{dt} r(t ; {\bm X}_0) = \Lambda_2 r(t ; {\bm X}_0),
	\end{align}
	for arbitrary ${\bm X}_0 \in {\mathbb R}^N$.
\end{Lemma}

\begin{proof}
	From Lemma 2, 
	\begin{align}
		&\frac{d}{dt}r(t ; {\bm X}_0) = \frac{d}{dt}\mathbb{E}[  \overline{Q_2}(S_{st}^t {\bm X}_0)  ] 
		= \Lambda_2 \mathbb{E}[  \overline{Q_2}(S_{st}^t {\bm X}_0)  ] = \Lambda_2 r(t ; {\bm X}_0).
	\end{align}
\end{proof}

Thus, the~averaged amplitude $r(t ; {\bm X}_0)$ obeys
\begin{align}
	\dot{r}(t ; {\bm X}_0) = \Lambda_2 r(t ; {\bm X}_0)
\end{align}
and hence
$	r(t ; {\bm X}_0) = e^{\Lambda_2 t} r(0 ; {\bm X}_0) $
as expected for any ${\bm X}_0$. Since $\Lambda_2$ is real and 
negative
by assumption, $r(t ; {\bm X}_0)$ decays exponentially with time.
This result indicates that the definition of the amplitude in Equation~(\ref{eq:TLamp}) for stochastic oscillators is a natural extension of the definition in Equation~(\ref{eq:koopmanamp}) for the deterministic systems from the Koopman-operator viewpoint.
As we illustrate in the next section with a few examples, the~above definition yields an amplitude value that decays linearly with $t$ on average and characterizes the deviation of the system state from the steady~oscillation.

%%%%%%%%%%%%%%%%%%%%%%%%%%%%%%%%%%%%%%%%%%%

\subsection{Limit of Vanishing Noise~Intensity}

Before proceeding to examples, we point out that the results for the stochastic oscillators formally reduce to the results for deterministic limit-cycle oscillators in the limit of vanishingly small~noise.

If we assume that the noise does not exist, i.e.,~${\bm D}({\bm X}) \to 0$ in the forward and backward Fokker--Planck Equations~(\ref{eq:fpe}) and (\ref{bkfpe}), we obtain the forward and backward Liouville equations~\cite{gaspard2005chaos,gardiner2009stochastic,lasota2013chaos,pavliotis2014stochastic},
\begin{align}
	\frac{\pa}{\pa t} p({\bm X}, t | {\bm Y}, s) 
	&= {\cal L}_{\bm X} p({\bm X}, t  | {\bm Y}, s),
	\cr
	\frac{\pa}{\pa s} p({\bm X}, t | {\bm Y}, s) &= - {\cal L}_{\bm Y}^{\dag} p({\bm X}, t | {\bm Y}, s),
\end{align}
where the forward Liouville operator is given by
\begin{align}
	{\cal L}_{\bm X} = - \frac{\pa}{\pa {\bm X}} {\bm A}({\bm X})
	\label{eq:clliouville}
\end{align}
and the backward Liouville operator is given by
\begin{align}
	{\cal L}_{\bm X}^{+} 
	= {\bm A}({\bm X}) \cdot \frac{\partial}{\partial {\bm X}}
	= {\bm A}({\bm X}) \cdot \nabla.
	\label{eq:bkclliouville}
\end{align}

Because the backward Liouville operator ${\cal L}_{\bm X}^{+}$ coincides with the infinitesimal generator of the Koopman operator ${A}$ in the deterministic case given in Equation~(\ref{koopmaninfini}), 
the Koopman eigenfunction $\Psi_0({\bm X})$ of $A$ in Equation~(\ref{eq:detiw}) 
is an eigenfunction of ${\cal L}_{\bm X}^{+}$ with an eigenvalue $i \omega$.
Thus, the~definition of the asymptotic phase for stochastic oscillators in Equation~(\ref{eq:TLphase2}) can be considered a natural generalization of the definition of the asymptotic phase for deterministic oscillators in Equation~(\ref{eq:koopmaniw}).
Similarly, 
the Koopman eigenfunction $R_0({\bm X})$ of $A$ in Equation~(\ref{eq:koopmanamp}) is an eigenfunction of ${\cal L}_{\bm X}^+$ with an eigenvalue $\Lambda_2 = \lambda_1$,
so the definition of the amplitude for stochastic oscillators in Equation~(\ref{eq:TLamp}) also corresponds to that for deterministic oscillators in Equation~(\ref{eq:koopmanamp}).

\section{Examples}
\unskip

\subsection{Numerical~Methods}

To demonstrate the validity of the phase and amplitude functions introduced in Section~\ref{sec:stochastic}, we consider two classical examples of noisy limit-cycle oscillators, i.e.,~the Stuart--Landau model~\cite{kuramoto1984chemical, nakao2016phase} and the FitzHugh--Nagumo model~\cite{fitzhugh1961impulses, nagumo1962active}.
We numerically calculate the eigenvalues and eigenfunctions of the backward Fokker--Planck operator and evaluate the phase and amplitude functions.
We also analyze a quantum limit-cycle oscillator in the semiclassical regime, which can be described by the same stochastic differential equations as those for the classical noisy limit-cycle~oscillators.

In the numerical calculations, we truncated the state space and approximated it as a finite square domain $ - D \leq x \leq D,  -D \leq y \leq D$ with a large enough value of $D$.
In all models considered, the~stationary PDF of the FPE rapidly decayed with the distance from the origin and took  numerically negligible values at the edges of the domain.
We discretized the domain into $N \times N$ grids and represented the PDF
as a $N^2$-dimensional vector.
We then represented the operator $L_{\bm X}^{+}$ as a $N^2 \times N^2$ matrix, calculated the eigenvalues and eigenvectors, and~obtained the phase and amplitude~functions.

\subsection{Example 1: Noisy Stuart--Landau~Model}

As the first example, we consider the Stuart--Landau model (the normal form of the supercritical Hopf bifurcation~\cite{guckenheimer1982nonlinear,kuramoto1984chemical}) under the effect of noise, described by
\begin{align}
	\label{eq:sl}
	dx &= \{ a x -  b y - ( cx - dy)(x^2+y^2)  \} dt + \sqrt{D_x} dW_x,
	\cr
	dy &= \{ b x + a y - ( dx + cy)(x^2+y^2)  \} dt  + \sqrt{D_y} dW_y,
\end{align}
where $x$ and $y$ are real variables, $a, b, c$, and~$d$ are real parameters, $W_x$ and $W_y$ are independent Wiener processes, and~$D_x$ and $D_y$ represent the intensities of the noise acting on $x$ and $y$, respectively.
The noiseless system with $D_x = D_y = 0$ has a stable limit cycle with a natural frequency $\omega = b - ad / c$ and the largest non-zero Floquet exponent $\lambda_1 = -2a$ when $a>0$ and $c>0$. For~this system, we can explicitly calculate the limit cycle and the phase and amplitude functions as~\cite{nakao2016phase}
\begin{align}
	&(x_0(\phi), y_0(\phi))^T = \sqrt{\frac{a}{c}} ( \cos \phi, \sin \phi )^T,
	\cr
	&\Phi_0(x, y) =  \tan^{-1} \left(\frac{y}{x}\right)-\frac{d}{c} \ln \sqrt{ \frac{ c}{a} (x^2 + y^2) },
	\cr
	&R_0(x, y) =  C_0 \left( c - \frac{a}{x^2 + y^2} \right).
	\label{deterministicSL}
\end{align}
where $C_0$ is an arbitrary scalar constant.
The basin $B_\chi$ of this limit cycle $\chi$ is the whole complex plane except the origin. 

In the following numerical simulations, we set the parameter values as $(a, b, c, d, D_x,$ $D_y)$ $= (0.5,$ $1.5$, $0.25$, $0.25$, $1$, $1)$. The~natural frequency and the largest non-zero Floquet exponent are $\omega = 1$ and $\lambda_1 = -1$, respectively.
It is noted that the fundamental frequency $\omega_1 = \mbox{Im}\ \Lambda_1$ and decay rate $\Lambda_2$ under the effect of noise are generally different from these deterministic values.
We used $D = 3.6$ and $N = 151$ for the numerical~analysis.

%\begin{figure}[H]
\begin{figure}[t]
	%\begin{center}
	\includegraphics[width=\hsize,keepaspectratio]{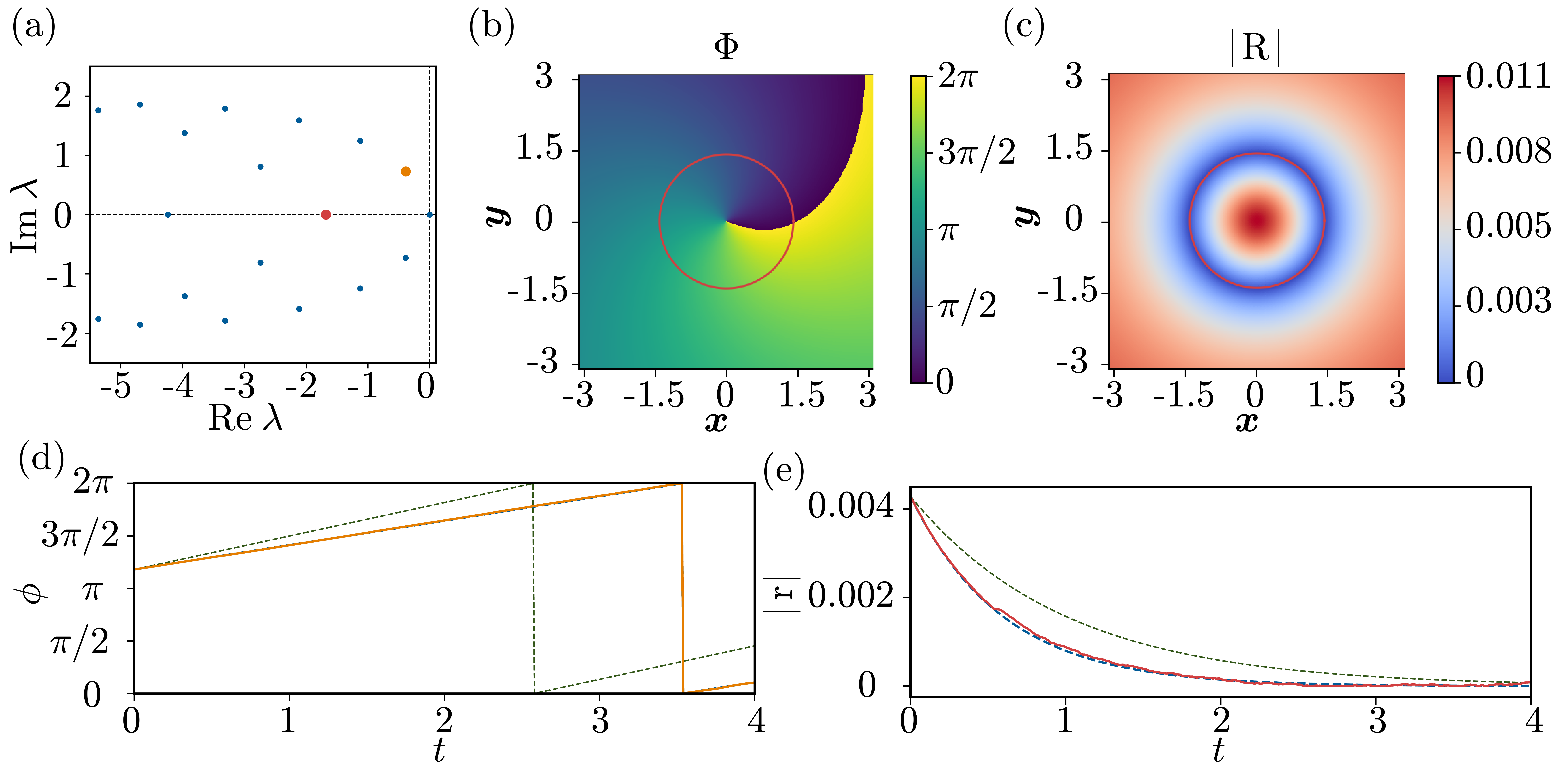}
	%\end{center}
	\caption{
		Phase and amplitude functions of a noisy Stuart--Landau model.
		(\textbf{a}) Eigenvalues of $L_{\bm X}^+$ near the imaginary axis. Orange and red dots represent $\Lambda_1$ and $\Lambda_2$, respectively.
		(\textbf{b}) Phase function $\Phi(x,y)$. The~phase origin is chosen as $(x,p)=(1.5, 0)$.
		(\textbf{c}) Amplitude function $\abs{R(x,y)}$.
		(\textbf{d}) Evolution of averaged phase $\phi$. 
		(\textbf{e}) Evolution of averaged amplitude $\abs{r}$.
		%%%
		In (\textbf{b},\textbf{c}), red-thin lines represent the deterministic limit-cycle solution.
		In (\textbf{d},\textbf{e}), averaged results over $10000$ trajectories (orange and red thin lines) and analytical solutions (blue-dotted lines) for the stochastic case and results for the deterministic case (green-dotted lines) are shown.
	}
	\label{fig_1}%MDPI: please change hythen into minus sign. and Figures should be close to where it is first mentioned in the text, we moved, please confirm.
	
\end{figure}
Figure~\ref{fig_1}a shows the eigenvalues of the Koopman operator $L_{\bm X}^+$ near the imaginary axis obtained numerically, where the eigenvalues $\Lambda_1 = \mu_1 + i \omega_1$ and $\Lambda_2$ are shown by orange and red dots, respectively.
The rightmost branch of the eigenvalues is approximately given by a parabola $\hat{\lambda}_n = i \omega_1 n - \mu_1 n^2~ (n=0, \pm 1, \pm 2, \ldots)$ passing through $\Lambda_1$~\cite{thomas2014asymptotic}.

Figure~\ref{fig_1}b,c show the phase function $\Phi(x,y)$ and amplitude function  $\abs{R(x,y)}$ 
associated with $\Lambda_1 = \mu_1 + i \omega_1$ and $\Lambda_2$, respectively.
We can observe that a circular region representing the local minima of the amplitude exists along the limit-cycle solution in the deterministic case and the phase increases from $0$ to $2\pi$ along this circle. 
In contrast to the deterministic case, Equation~(\ref{deterministicSL}), the~amplitude does not diverge at the unstable fixed point at the origin $(x, y) = (0, 0)$, because~the system state can escape from this point in a finite time due to the effect of~noise.

To confirm that these functions yield appropriate values of the phase and amplitude on average, we obtained $10000$ trajectories by direct numerical simulations of the\linebreak Equation~(\ref{eq:sl}) from the initial point $(x_0, y_0) = (-1.5, -1.5)$ and calculated the averaged phase $\phi = \mbox{\rm Arg}\  \big[ \overline{ Q_1 }(x,y) \big]$ and  amplitude 
$ \abs{r} =   \big[ \abs{\overline{ Q_2 }(x,y)} \big]$, where $\big[ \cdot \big]$ represents a sample average over all obtained trajectories.
Figure~\ref{fig_1}d,e show that these values are in good agreement with the analytical solutions $\phi = \omega_1 t + \phi_0$ and $\abs{r} = \abs{r_0} \exp(\Lambda_2 t)$, where the fundamental frequency $\omega_1 = 0.728$ and the decay rate $\Lambda_2 = -1.680$ are numerically evaluated from the eigenvalues plotted in Figure~\ref{fig_1}a.
For comparison, we also show the analytical solutions for the deterministic case without noise ($D_x = D_y = 0$), namely, 
$\phi = \omega t + \phi_0 = t + \phi_0$ and $\abs{r} = \abs{r_0} \exp( \lambda_1 t) = \abs{r_0} \exp(- t)$.
The averaged phase increases more slowly
and the averaged amplitude decays more quickly than those in the deterministic case due to the effect of~noise.

\subsection{Example 2: Noisy FitzHugh--Nagumo~Model}

Next, we consider the FitzHugh--Nagumo model~\cite{fitzhugh1961impulses, nagumo1962active}
subjected to noise, described~by
\begin{align}
	\label{eq:fhz}
	dx &= ( x - a_1x^3 - y ) dt + \sqrt{D_x} dW_x,
	\cr
	dy &= \eta_1 (x + b_1) dt + \sqrt{D_y} dW_y,
\end{align}
where $x$ and $y$ are real variables,
$a_1$, $b_1$, and~$\eta_1$ are real parameters, 
$W_x$ and $W_y$ are independent Wiener processes, and~$D_x$ and $D_y$ represent the intensities of the noise, respectively.

First, we consider parameter set (A): $(a_1, b_1,$ $\eta_1, D_x, D_y)$ $= (1/3, 0.5, 0.5, 0.2, 0.2)$,
at which the deterministic vector field possesses a stable limit-cycle solution with
the natural frequency $ \omega = 0.588$ and the largest non-zero Floquet exponent $\lambda_1 = -1.11$.
We used $D = 4.2$ and $N = 151$ for the numerical~analysis.

%\begin{figure}[H]
\begin{figure}[t]
	%\begin{center}
	\includegraphics[width=\hsize,keepaspectratio]{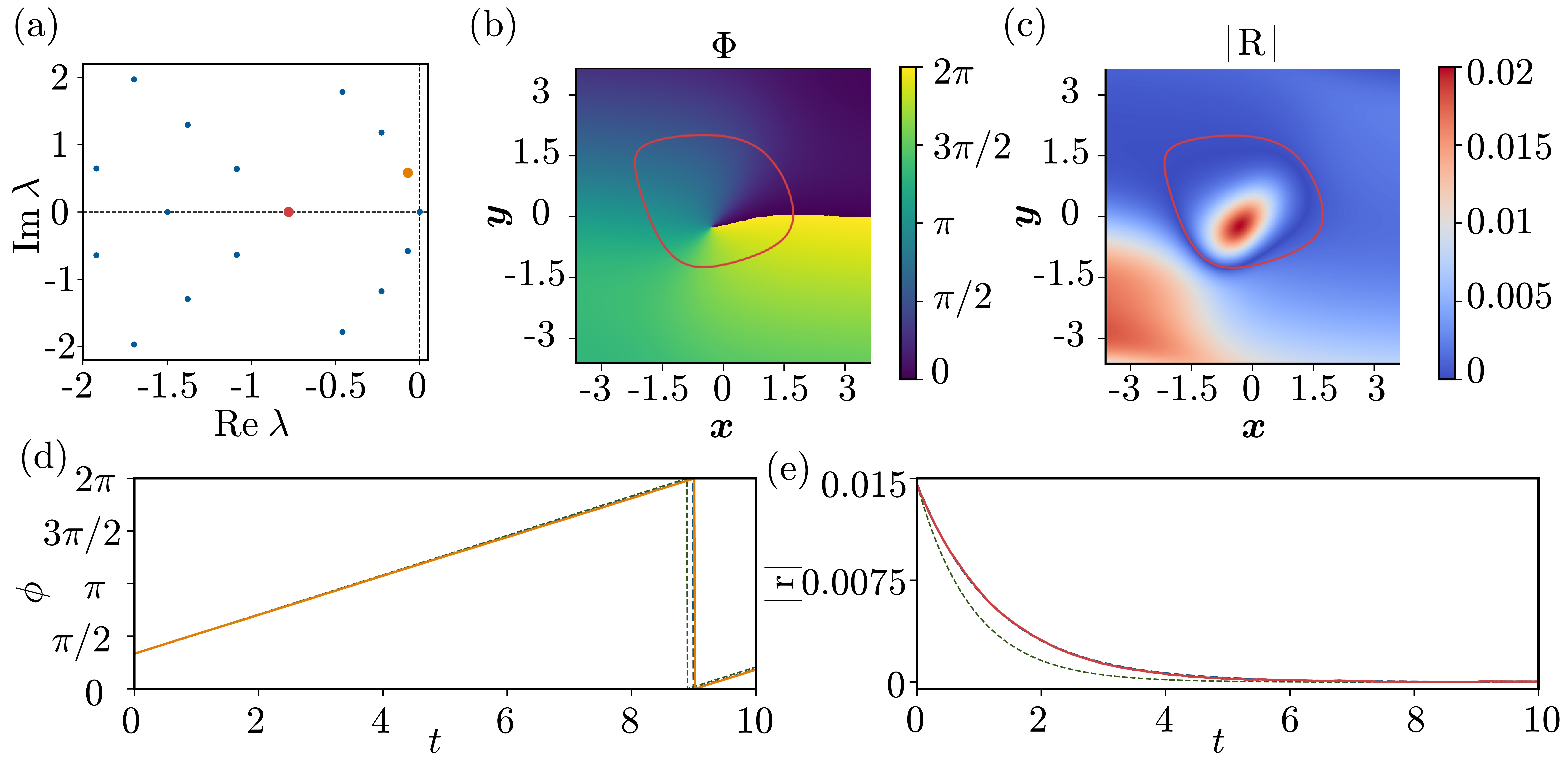}
	%\end{center}
	\caption{
		Phase and amplitude functions of a noisy FitzHugh--Nagumo model with parameter set (A). The~deterministic vector field possesses a limit-cycle solution.
		(\textbf{a}) Eigenvalues of $L_{\bm X}^+$ near the imaginary axis. Orange and red dots represent $\Lambda_1$ and $\Lambda_2$, respectively.
		(\textbf{b}) Phase function $\Phi$.
		(\textbf{c}) Amplitude function $\abs{R}$.
		%%%
		(\textbf{d}) Evolution of averaged phase $\phi$. (\textbf{d}) Evolution of averaged amplitude $\abs{r}$.
		In (\textbf{b}),  $(x,p)=(1.5, 0)$ is chosen as the phase origin. In~(\textbf{b},\textbf{c}), red-thin lines represent the deterministic limit-cycle solution.
		In (\textbf{d},\textbf{e}), averaged results over 10,000 trajectories (orange and red thin lines) and analytical solutions (blue-dotted lines) for the stochastic case and results for the deterministic case (green-dotted lines) are shown.
	}
	\label{fig_2}
\end{figure}

Figure~\ref{fig_2}a shows the eigenvalues of the Koopman operator ${L}_{\bm X}^+$ near the imaginary axis obtained numerically, where  $\Lambda_1 = \mu_1 + i \omega_1$ and $\Lambda_2$ are shown by orange and red dots, respectively.
The rightmost branch of the eigenvalues is approximately a parabola $\hat{\lambda}_n = i \omega_1 n - \mu_1 n^2~ (n=0, \pm 1, \pm 2, \ldots)$ passing through $\Lambda_1$, which is qualitatively similar to the one for the noisy Stuart--Landau model in Figure~\ref{fig_1}a.
Figure~\ref{fig_2}b,c show the phase $\Phi(x,y)$ and amplitude  $\abs{R(x,y)}$ associated with $\Lambda_1 = \mu_1 + i \omega_1$ and $\Lambda_2$.
As in the case of the noisy Stuart--Landau model, a~circular region corresponding to the local minima of the amplitude function exists around the deterministic limit-cycle solution and the phase increases along this region. The~amplitude does not diverge at the unstable fixed point due to the effect of~noise.

We calculated the time evolution of $\Phi(x, y)$ and $\abs{R(x, y)}$ by direct numerical simulations of Equation~(\ref{eq:fhz}) %~(\ref{eq:sl}) 
from the initial point $(x_0, y_0) = (0.1, 0.1)$ and averaged the results over 10,000 trajectories.
Figure~\ref{fig_2}d,e show the averaged phase $\phi = \mbox{\rm Arg}\  \big[ \overline{ Q_1 }(x,y) \big]$ and amplitude 
$\abs{r} =   \big[ { \abs{ \overline{Q_2} (x,y)}} \big]$.
They are in good agreement with the analytical solutions $\phi = \omega_1 t + \phi_0$ and $\abs{r} = \abs{r_0} \exp(\Lambda_2 t)$, where the fundamental frequency $\omega_1 = 0.582$ and the decay rate $\Lambda_2 = -0.778$ are numerically evaluated from the eigenvalues plotted in Figure~\ref{fig_2}a.
For comparison, we also show the analytical solution for the deterministic case without noise ($D_x = D_y = 0$), namely, $\phi = \omega t + \phi_0 = 0.588t + \phi_0$ and $\abs{r} = \abs{r_0} \exp(\lambda_1 t) = \abs{r_0} \exp(-1.11  t)$.
The phase evolves more slowly
and also the amplitude decays more slowly than those in the deterministic case due to the effect of~noise.

Next, we consider parameter set (B): $(a_1, b_1,$ $\eta_1, D_x, D_y) = (1/3, 1.05, 0.25, 0.1, 0.1)$. 
In this case, the~deterministic vector field does not have a stable limit-cycle, but~the system is close to a supercritical Hopf bifurcation of a limit cycle. Thus, relatively regular noise-induced oscillations occur even though the system does not have a deterministic limit cycle, a~phenomenon known as the coherence resonance~\cite{gang1993stochastic, pikovsky1997coherence, lindner2000coherence}.
We used $D = 3.9$ and $N = 151$ for the numerical~analysis.

%\begin{figure}[H]
\begin{figure}[t]
	%\begin{center}
	\includegraphics[width=\hsize,keepaspectratio]{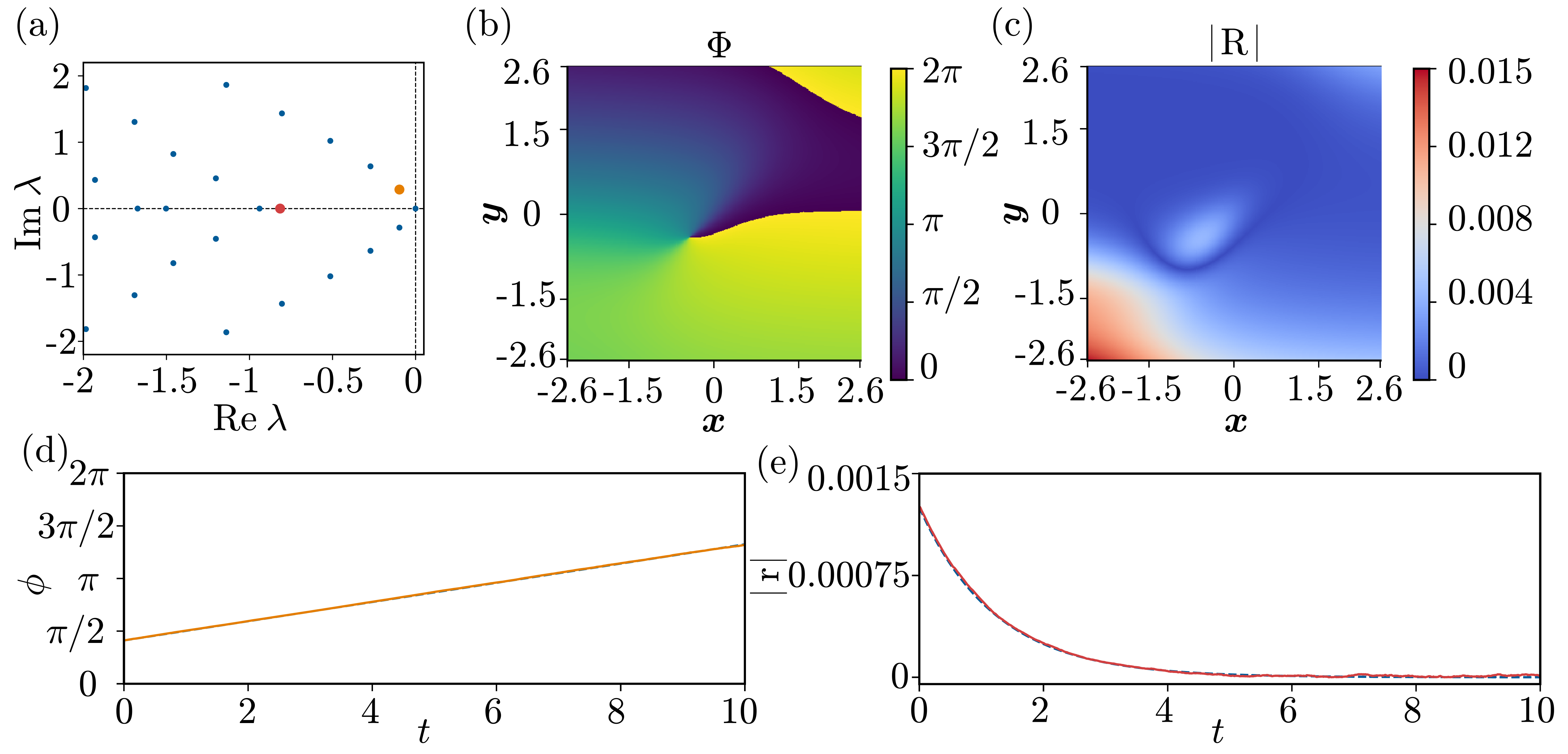}
	%\end{center}
	\caption{
		Phase and amplitude functions of a noisy FitzHugh--Nagumo model with parameter set (B). The~deterministic vector field does not possess a limit cycle, but regular noise-induced oscillations occur.
		(\textbf{a}) Eigenvalues of ${L}_{\bm X}^+$ near the imaginary axis. Orange and red dots represent $\Lambda_1$ and $\Lambda_2$, respectively.
		(\textbf{b}) Phase function $\Phi$. The~phase origin is chosen as $(x,p)=(1.5, 0)$.
		(\textbf{c}) Amplitude function $\abs{R}$.
		(\textbf{d}) Evolution of averaged phase $\phi$.
		(\textbf{e}) Evolution of averaged amplitude $\abs{r}$.
		In (\textbf{d},\textbf{e}), averaged results over 10,000 trajectories (orange-thin or red-thin lines) and analytical solutions (blue-dotted lines) are shown.
	}
	\label{fig_3}
\end{figure}

Figure~\ref{fig_3}a shows the eigenvalues of the Koopman operator ${L}_{\bm X}^+$ near the imaginary axis obtained numerically, where $\Lambda_1 = \mu_1 + i \omega_1$ and $\Lambda_2$ are shown by orange and red dots, respectively.
Figure~\ref{fig_3}b,c show the phase $\Phi(x,y)$ and amplitude  $\abs{R(x,y)}$ associated with $\Lambda_1 = \mu_1 + i \omega_1$ and $\Lambda_2$, respectively.
Interestingly, although~the deterministic system does not have a limit-cycle solution, we can still observe in Figure~\ref{fig_3}b,c a circular region representing the local minima of the amplitude function.
This region corresponds to the noise-induced oscillations and the phase increases along this circular~region.

Figure~\ref{fig_3}d,e show the time evolution of the average values of the phase and amplitude,  which are averaged over 10,000 trajectories by direct numerical simulations of the %\mbox{Equation~(\ref{eq:sl})} 
\mbox{Equation~(\ref{eq:fhz})} from the initial point $(x_0, y_0) = (0.1, 0.1)$.
The averaged phase $\phi = \mbox{\rm Arg}\  \big[ \overline{ Q_1 }(x,y) \big]$ and amplitude $r =  \big[ \abs{ \overline{ Q_2 }(x,y)} \big]$ show good agreement with the analytical solutions $\phi = \omega_1 t + \phi_0$ and $\abs{r} = \abs{r_0} \exp(\Lambda_2 t)$, where the fundamental frequency $\omega_1 = 0.287$ and the decay rate $\Lambda_2 = -0.815$ are numerically evaluated from the eigenvalues in Figure~\ref{fig_3}a.

Thus, we can introduce the phase and amplitude functions also in this case without a deterministic limit cycle by using the present definition using the Koopman~eigenfunctions.

\subsection{Example 3: Semiclassical Stuart--Landau~Model} %Oscillator}   

Finally, we apply the proposed definition of the stochastic phase and amplitude functions to a quantum limit-cycle oscillator in the semiclassical regime. As~an example, 
we use the quantum Stuart--Landau model~\cite{chia2020relaxation, arosh2021quantum} (also known as the quantum van der Pol model~\cite{lee2013quantum}) with a Kerr effect~\cite{lorch2016genuine, kato2019semiclassical} in quantum~optics.

Employing the phase space approach~\cite{gardiner1991quantum, carmichael2007statistical}, the~system state can be represented by a Wigner function $W(x, y)$ (see Reference~\cite{lorch2016genuine, kato2019semiclassical} for details). 
In the semiclassical regime, the~quantum noise is sufficiently weak and $W(x, y)$ approximately obeys a quantum FPE, which has the same form as the ordinary FPE for classical systems. Thus, we can derive the corresponding Ito SDE from the quantum FPE as
\begin{align}
	\label{eq:qvdp_ldv_re}
	d
	\left( \begin{matrix}
		x  \\
		y  \\
	\end{matrix} \right)
	&=	
	\left( \begin{matrix}
		\frac{\gamma_1  + 2\gamma_2}{2} x - (\omega_0 + 2K) p - ( \gamma_2 x - 2 K p )( x^2 + p^2 )  \\
		(\omega_0 + 2K) x + \frac{\gamma_1  + 2\gamma_2}{2} p - ( 2 K x + \gamma_2 p )( x^2 + p^2 ) \\
	\end{matrix} \right) dt
	+ 
	\sqrt{  \frac{\beta(x, y)}{2}}
	\left( \begin{matrix}
		dW_x
		\\
		dW_y
		\\
	\end{matrix}
	\right),
\end{align}
where $\beta(x, y)  = \frac{\gamma_{1}}{2}+2 \gamma_{2}\left(x^2 + y^2 -\frac{1}{2}\right)$,
$\omega_{0}$ is a frequency parameter of the oscillator, $K$ represents the Kerr parameter,
and $\gamma_{1}$ and $\gamma_{2}$ are the decay rates for the negative damping and nonlinear damping, respectively, and~$W_x$ and  $W_y$ are independent Wiener processes. The~semiclassical approximation is valid when $\gamma_{2}$ and $K$ are sufficiently small~\cite{lorch2016genuine, kato2019semiclassical}.

The deterministic part of Equation~(\ref{eq:qvdp_ldv_re}) is equivalent to the Stuart--Landau model used in Example 1
and only the coefficient of the noise term differs. Therefore, the~deterministic limit cycle and the phase and amplitude functions can be obtained from the results for Equation~(\ref{deterministicSL}), where the parameters are given by $a=\frac{\gamma_1  + 2\gamma_2}{2}$, $b=\omega_0 + 2K, c=  \gamma_2, d = 2 K$.
We set the parameter values as $( \gamma_1, \gamma_{2}, \omega_0, K) = (1, 0.05, 1, 0.025)$. With~these values, %the~system is in the semiclassical regime and 
the deterministic vector field of Equation~(\ref{eq:qvdp_ldv_re}) possesses a stable limit-cycle solution with the natural frequency $ \omega = \omega_0 -  K \gamma_1 /\gamma_2 = 0.5$ and the largest non-zero Floquet exponent $\lambda_1 =  \gamma_1  + 2\gamma_2 = -1.1$. We used $D = 6.7$ and $N = 151$ for the numerical~analysis.

Figure~\ref{fig_4}a shows the eigenvalues of the Koopman operator $L_{\bm X}^+$, where $\Lambda_1 = \mu_1 + i \omega_1$ and $\Lambda_2$ are indicated by orange and red dots, respectively.
Figure~\ref{fig_4}b,c show the phase function $\Phi(x,y)$ and amplitude function $R(x,y)$ associated with $\Lambda_1 = \mu_1 + i \omega_1$ and $\Lambda_2$.
For comparison, we also show in Figure~\ref{fig_4}d the deterministic phase function $\Phi_0(x,y)$ in Equation~(\ref{deterministicSL}).
The stochastic phase function $\Phi(x,y)$ in Figure~\ref{fig_4}b is slightly different from the deterministic phase function $\Phi_0(x,y)$ in Figure~\ref{fig_4}d, in~particular near the origin, because~the stochastic phase function includes the effect of weak quantum noise.
Figure~\ref{fig_4}c shows that the amplitude function takes the minima around the deterministic limit cycle. 
We note that the amplitude function does not diverge at the origin in contrast to the deterministic case.
Here, we used a color map with the maximum value of $0.005$ to enhance the local minima of the amplitude function, and~the region near the origin where the amplitude is larger than $0.005$ is shown in the same color (the true maximum amplitude is $0.0876$ at the origin).

%\begin{figure}[H]
\begin{figure}[t]
	%\begin{center}
	\includegraphics[width=\hsize,keepaspectratio]{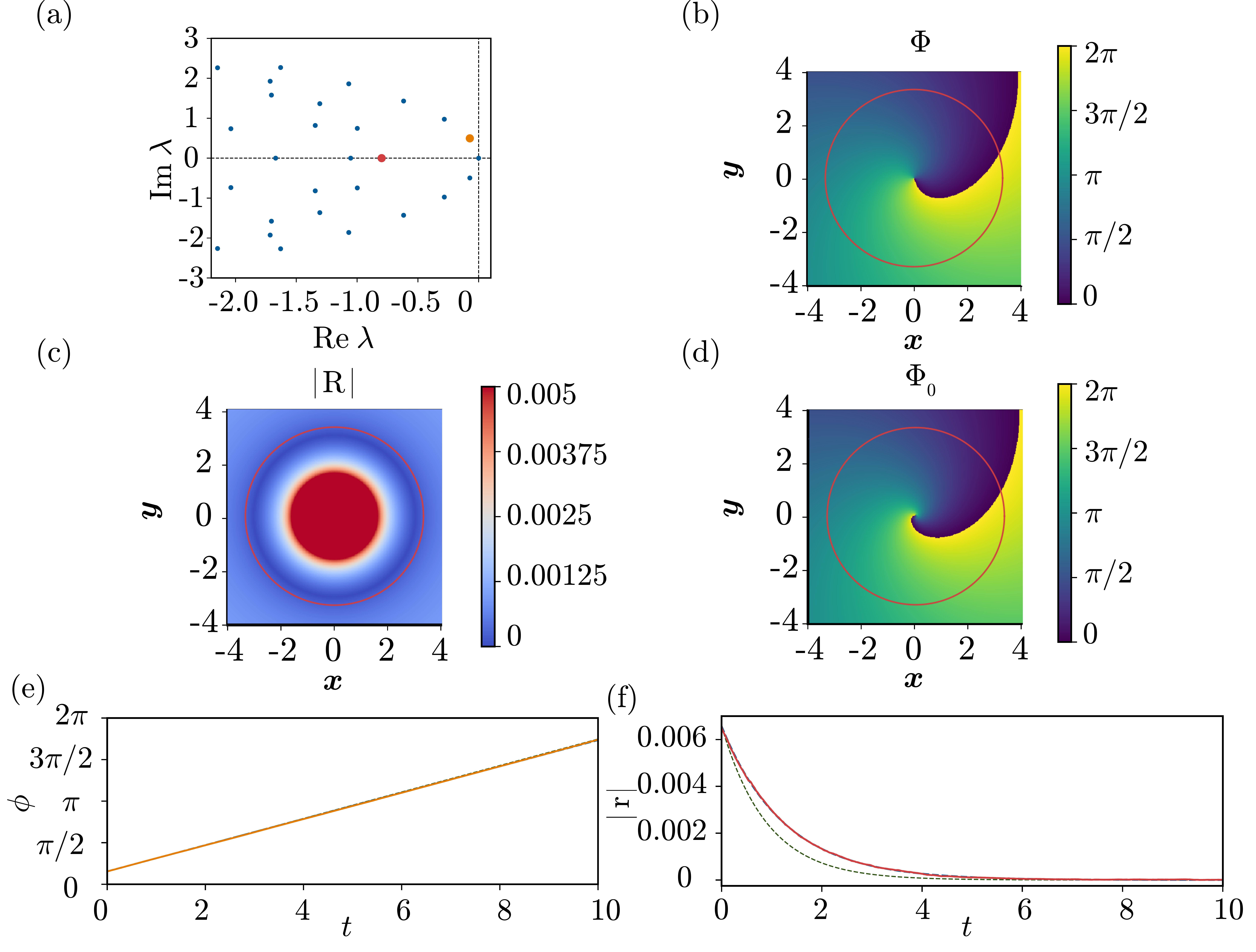}
	%\end{center}
	\caption{
		Phase and amplitude functions of a quantum Stuart--Landau model
		with a quantum Kerr effect in the semiclassical regime.
		(\textbf{a}) Eigenvalues of $L_{\bm X}^+$ near the imaginary axis. Orange and red dots represent eigenvalues $\Lambda_1$ and $\Lambda_2$, respectively.
		(\textbf{b}) Phase function $\Phi$.
		(\textbf{c}) Amplitude function $\abs{R}$.
		(\textbf{d}) Phase function $\Phi_0$ of the deterministic system.
		%%%
		(\textbf{e}) Evolution of averaged phase $\phi$. (\textbf{f}) Evolution of averaged amplitude $\abs{r}$.
		In (\textbf{b},\textbf{d}), $(x,p)=(2.5, 0)$ is chosen as the phase origin. % and the red-thin lines represent the deterministic limit-cycle solutions.
		In (\textbf{e},\textbf{f}), averaged results over 10,000 trajectories (orange and red thin lines) and analytical solutions for the semiclassical case (blue-dotted lines) and results for the deterministic case (green-dotted lines) are shown.
	}
	\label{fig_4}
\end{figure}

Figure~\ref{fig_4}e,f show the time evolution of the values of the phase and amplitude averaged over 10,000 trajectories obtained by direct numerical simulations of the semiclassical Equation~(\ref{eq:sl}) from the initial point $(x_0, y_0) = (2.5, 0)$.
They show good agreement with the analytical solutions $\phi = \omega_1 t + \phi_0$ and $\abs{r} =\abs{ r_0} \exp(\Lambda_2 t)$, where the values of the fundamental frequency $\omega_1 = 0.496$ and the decay rate $\Lambda_2 = -0.798$ are numerically evaluated from the eigenvalues shown in Figure~\ref{fig_4}a.

Thus, the~present definition of the phase and amplitude functions is also applicable to a quantum Stuart--Landau model
in the semiclassical regime and yields reasonable~values.

%%%%%%%%%%%%%%%%%%%%%%%%%%%%%%%%%%%%%%%%%%

\section{Discussion}

We proposed a definition of the asymptotic phase and amplitude functions for stochastic oscillatory systems by generalizing the definition for deterministic limit-cycle oscillators on the basis of the Koopman operator theory, motivated by the definition of the stochastic asymptotic phase introduced by Thomas and Lindner~\cite{thomas2014asymptotic}. 

The proposed asymptotic phase and amplitude for strongly stochastic oscillatory systems may be used for systematic and quantitative analysis of synchronization phenomena in stochastic oscillators.
We may also be able to develop a phase-amplitude reduction theory for strongly stochastic oscillators by using the phase and amplitude functions, which allows us to reduce the system dynamics subjected to weak external inputs to a simple two-dimensional set of equations.
Such theories will facilitate detailed analysis, control, and~optimization of noise-induced oscillatory phenomena, including the
coherence-resonance and self-induced-stochastic-resonance oscillations 
~\cite{zhu2020phase,zhu2021stochastic}. 
%have been studied recently in References

It will also be interesting to introduce amplitude functions for strongly quantum oscillatory systems on the basis of the Koopman operator theory \cite{kato2020quantum}. %~\cite{kato2019semiclassical, kato2020quantum}. 
We have recently defined the quantum asymptotic phase of strongly quantum oscillators by using the eigenoperator of the adjoint Liouville superoperator and found that it can largely differ from the asymptotic phase in the classical limit~\cite{kato2020quantum}. 
The amplitude function, which can be defined similarly, may also be very different from the classical counterpart and characterize the quantum signatures of synchronization observed in the strong quantum~regime.

%%%%%%%%%%%%%%%%%%%%%%%%%%%%%%%%%%%%%%%%%%

\section{Conclusions}

We proposed a definition of the asymptotic phase and amplitude functions for stochastic oscillatory systems. %by generalizing the definition for deterministic limit-cycle oscillators on the basis of the Koopman operator theory. 
The proposed phase and amplitude functions are introduced in terms of the backward Fokker--Planck operator, which can be interpreted as the Koopman operator for the stochastic system.
The validity of the phase and amplitude functions was numerically demonstrated for noisy Stuart--Landau and FitzHugh--Nagumo models
and also for a quantum Stuart--Landau model
in the semiclassical~regime.\\

\textit{Note added}—
During the preparation of this article
, we noticed a new, closely related study by P\'erez-Cervera, Lindner, and~Thomas~\cite{perez2021isostables}, which introduced the isostables (level sets of the asymptotic amplitude function) for stochastic oscillators on the basis of the backward Kolmogorov equation and analyzed the examples of a spiral sink, a~noisy Stuart--Landau oscillator, and~a noisy heteroclinic oscillator. Our present results differ from Reference~\cite{perez2021isostables} in that we explicitly discussed the relationship with the Koopman operator theory and analyzed a noisy excitable system and quantum limit-cycle oscillator in the semiclassical regime, in~addition to noisy limit-cycle oscillators. We thus believe our results provide different insights %into the present problem
and are complementary to Reference~\cite{perez2021isostables}.

\vspace{6pt} 

{This research was funded by JSPS KAKENHI JP17H03279, JP18H03287, JPJSBP120202201, JP20J13778, JP20F40017, JP20K19883, JP19KT0030, JP19K03671, and JST CREST JP-MJCR1913}
\\

%\bibliographystyle{apsrev4-1}
%\bibliography{kato-nakao-sto_iso.bib}

%merlin.mbs apsrev4-1.bst 2010-07-25 4.21a (PWD, AO, DPC) hacked
%Control: key (0)
%Control: author (72) initials jnrlst
%Control: editor formatted (1) identically to author
%Control: production of article title (-1) disabled
%Control: page (0) single
%Control: year (1) truncated
%Control: production of eprint (0) enabled
%

\end{document}